\documentclass[12pt]{article}
\usepackage{amsmath, amssymb, amsthm}
\usepackage[hypcap=true]{caption}
\usepackage{comment}
\usepackage{enumitem}
\usepackage{float}
\usepackage{fullpage}
\usepackage{hyperref}
\usepackage{mathtools}
\usepackage{soul}
\usepackage{subcaption}
\usepackage{titlesec}

\theoremstyle{plain}
\newtheorem{theorem}{Theorem}[section]
\newtheorem{corollary}{Corollary}[section]
\newtheorem{lemma}{Lemma}[section]
\newtheorem{proposition}{Proposition}[section]

\theoremstyle{definition}

\theoremstyle{remark}

\newcommand{\m}{\mathcal{M}} 
 
\newcommand{\ot}{\widetilde{O}} 

\DeclarePairedDelimiter\abs{\lvert}{\rvert}
\DeclarePairedDelimiter\set{\{}{\}}
\DeclarePairedDelimiter\parens{(}{)}

\DeclarePairedDelimiter\floor{\lfloor}{\rfloor}

\newcommand{\bb}[1]{\ensuremath{\mathbb{#1}}}

\newcommand{\R}{\bb{R}}

\newcommand{\Z}{\bb{Z}}

\newcommand{\setto}[2][1]{\set{#1,\ldots,#2}}

\begin{document}

\title{Approximately Sampling Elements with Fixed Rank in Graded Posets}
\author{
  Prateek Bhakta
  \footnote{
    Department of Math and Computer Science, University of Richmond, Richmond, VA 23173.
    Email: \href{mailto:pbhakta@richmond.edu}{\texttt{pbhakta@richmond.edu}}.
    Supported in part by NSF grant CCF-1526900.
  }
	\and Ben Cousins
  \footnote{
    School of Computer Science, Georgia Institute of Technology, Atlanta, GA 30332.
    Email: \href{mailto:bcousins3@gatech.edu}{\texttt{bcousins3@gatech.edu}}.
    Supported in part by NSF grants CCF-1217793 and EAGER-1415498.
  }
  \and Matthew Fahrbach
  \thanks{
    School of Computer Science, Georgia Institute of Technology, Atlanta, GA 30332.
    Email: \href{mailto:matthew.fahrbach@gatech.edu}{\texttt{matthew.fahrbach@gatech.edu}}.
    Supported in part by NSF grant DGE-1650044 and a Tau Beta Pi Fellowship.
  }
  \and Dana Randall
	\thanks{
    School of Computer Science, Georgia Institute of Technology, Atlanta, GA 30332.
    Email: \href{mailto:randall@cc.gatech.edu}{\texttt{randall@cc.gatech.edu}}.
    Supported in part by NSF grants CCF-1526900 and CNS-1544090.
  }
}
\date{\today}
\maketitle

\begin{abstract} 
%
%
Graded posets frequently arise throughout combinatorics, where it is natural
to try to count the number of elements of a fixed rank.
These counting problems are often \#\textbf{P}-complete, so we consider
approximation algorithms for counting and uniform sampling.
We show that for certain classes of posets, biased
Markov chains that walk along edges of their Hasse diagrams allow us to
approximately generate samples with any fixed rank in expected polynomial time.
Our arguments do not rely on the typical proofs of log-concavity, which are
used to construct a stationary distribution with a specific mode in order to
give a lower bound on the probability of outputting an element of the desired
rank.
Instead,
we infer this directly from bounds on the mixing time of the chains through a
method we call \textit{balanced bias}.

A noteworthy application of our method is sampling restricted classes of
integer partitions of $n$.  We give the first provably efficient Markov chain
algorithm to uniformly sample integer partitions of $n$ from general restricted
classes.
Several observations allow us to improve the efficiency of this
chain to require $O(n^{1/2}\log(n))$ space, and for
unrestricted integer partitions, expected $O(n^{9/4})$ time.
Related applications include sampling permutations with a fixed number of
inversions and lozenge tilings on the triangular lattice with a fixed average
height.
\end{abstract}

\newpage

\section{Introduction}\label{sec:intro}
Graded posets are partially ordered sets equipped with a unique rank function
that both respects the partial order and such that neighboring elements in the
{\it Hasse diagram} of the poset have ranks that differ by $\pm 1$.  Graded
posets arise throughout combinatorics, including permutations ordered by
numbers of inversions, geometric lattices ordered by volume, and independent
sets and matchings ordered by cardinality.
Sometimes we find rich underlying structures that allow us to directly count,
and therefore sample, fixed rank elements of a graded poset. In other cases,
efficient methods are unlikely to exist, so Markov chains offer the best
approach to sampling and approximate counting.

Jerrum and Sinclair \cite{JS-book} observed that we could sample matchings of any fixed
size with the addition of a bias parameter $\lambda$ that gives weight
proportional to $\lambda^{|m|}$ to each matching~$m$.  For any graph $G$, they
showed that the sequence $a_i$, the number of matchings of $G$ of size~$i$, is
log-concave, from which it follows that $f(i) = a_i \lambda^i$ is also. In
particular,  $f(i)$ must be unimodal for all $\lambda$.  Setting $\lambda =
a_k/a_{k+1}$ makes $k$ the mode of distribution $f(i)$, and therefore samples
with this weighting will be of the appropriate size with probability at least
$1/(n+1)$.  Jerrum and Sinclair showed that the matching Markov chain  is
rapidly mixing for all $\lambda$, so it can find matchings of fixed size $k$
efficiently whenever $1/{\rm poly}(n) < \lambda < {\rm poly}(n)$.  (This
condition is not always satisfied, but the more involved algorithm of Jerrum,
Sinclair, and Vigoda circumvents this issue \cite{JSV06}.) Log-concavity is
critical to this argument in order to conclude that there is a value of
$\lambda$ for which samples of the desired size occur with high enough
probability.

This follows a common approach used in physics for which we would like to
sample from a {\it microcanonical ensemble}, i.e., the states with a fixed
energy, from a much larger {\it canonical} (or
{\it grand canonical}) {\it ensemble}, where the energies are allowed to vary
due to interactions with the external environment. In particular, given input
parameter $\lambda$, often related to temperature, a configuration $\sigma$ has
{\it Gibbs} (or {\it Boltzmann}) weight  $\pi(\sigma) = \lambda^{r(\sigma)}/Z$,
where $r(\sigma)$ is the rank of $\sigma$ and $Z$ is the normalizing constant.
Elements $\sigma$ sampled from this distribution are  uniformly distributed,
conditioned on their rank.  The choice of $\lambda$ controls the expected rank
of the distribution, so simulations of the Markov chain at various $\lambda$ 
can be useful for understanding properties of
configurations with a fixed energy. Typically, however, there is no a priori
guarantee that this approach will enable us to sample configurations of a given
size efficiently.

Our main example throughout will be sampling and counting (possibly restricted) {\it integer partitions}.  An integer partition of nonnegative integer $n$ is a decomposition
of $n$ into a nonincreasing sequence of positive integers that sum to $n$. The seven
partitions of  $5$ are: $(5)$, $(4, 1)$, $(3, 2)$, $(3, 1, 1)$, $(2,
2, 1)$, $(2, 1, 1, 1)$, and $(1, 1, 1, 1, 1)$.  Integer partitions are commonly represented
by staircase walks in $\Z^2$ known as \textit{Young} (or \textit{Ferrers}) \textit{diagrams}, where
the heights of the columns represent distinct pieces of the partition.
Partitions of $n$  have exactly $n$ squares,
i.e., the \emph{area of the diagram}, and their column heights are nonincreasing.
Partitions arise in many contexts, include exclusion 
processes~\cite{cmo}, random
matrices~\cite{ok}, representation theory~\cite{james1984representation},
juggling patterns \cite{abcn}, and growth processes~\cite{GPR09}
(see,~e.g.,~\cite{andrews}).

\subsection{Sampling Elements from Graded Posets}
Several general approaches have been developed to sample elements of fixed
rank from a graded poset, with varying success.
The three main approaches for sampling are dynamic programming algorithms using 
self-reducibility, Boltzmann samplers using geometric random variables, and
Markov chains.  The first two approaches require methods to estimate the
number of configurations of each size, so Markov chains offer the most
promising approach for sampling when these are unavailable.

Each of these approaches has been studied extensively in the context of sampling
{integer partitions}.
The first class of approaches uses dynamic programming and generating
functions to iteratively count the number of partitions of a given type.
Nijinhuis and Wilf~\cite{nw} give a recursive algorithm using dynamic
programming that computes tables of exact values.
This algorithm takes $O(n^{5/2})$ time and space for preprocessing and
$O(n^{3/2})$ time per sample.  Squire~\cite{squire} improved this to $O(n^2)$
time and space for preprocessing and $O(n^{3/2} \log(n))$ time per sample 
using Euler's pentagonal recurrence and a more efficient search method.  The
time and space complexity bounds of these algorithms account for the fact that
each value of $p(n)$, as well as the intermediate summands, requires
$O(n^{1/2})$ space by the Hardy-Ramanujan formula. Therefore, even when available, dynamic
programming approaches for exact sampling break down in practice on single
machines when~${n\ge 10^6}$ due to space constraints.

\emph{Boltzmann samplers} offer a more direct method for sampling that avoids the
computationally expensive task of counting partitions.  A
Boltzmann sampler generates samples from a larger combinatorial class with
probability proportional to the Boltzmann weight~$\lambda^{\abs{\sigma}}$,
where $\abs{\sigma}$ is the size of the partition.  Samples of the same size
are drawn uniformly at random, and the algorithm rejects those that fall
outside of the target size~\cite{dfls, ffp}.  The value $\lambda$ is chosen to
maximize the yield of samples of our target size $n$.
Fristedt~\cite{fristedt} suggested an approach that quickly generates a
random partition using appropriate independent geometric random variables.  His
approach exploits the factorization of the generating function for $p(n)$ and
can be interpreted as sampling Young diagrams $\sigma$ in the $n \times
\infty$ grid with probability proportional to the Boltzmann
weight~$\lambda^{\abs{\sigma}}$.  Recently Arratia and
DeSalvo~\cite{ad} gave a probabilistic approach that is substantially more
efficient than previous algorithms, thus allowing for fast generation of random
partitions for significantly larger numbers, e.g., $n \ge 10^6$.  Building on
the work of Fristedt~\cite{fristedt}, they introduce the \textit{probabilistic
divide-and-conquer} (PDC) method to generate random partitions of $n$ in
optimal $\ot(n^{1/2})$ expected time and space (where $\ot$ suppresses $\log$
factors). Their PDC algorithm also uses independent geometric
random variables to generate a partition, but does so recursively in phases.
PDC achieves superior performance relative to conventional Boltzmann Sampling
by rejecting impossible configurations in early phases.

Stochastic approaches using Markov chains have produced a similarly
rich corpus of work, but until now have not provided rigorous
polynomial bounds. One popular direction uses Markov chains based on
\emph{coagulation and fragmentation processes} that allow pieces of the
partition to be merged and split~\cite{aldous, bp}.  Ayyer et al.~\cite{abcn}
recently proposed several natural Markov chains on integer partitions in order
to study juggling patterns. In all of these works, most of the effort has been
to show that the Markov chains converge to the uniform distribution over
partitions and often use stopping rules in order to generate samples.
Experimental evidence suggests that these chains may converge quickly to
the correct equilibrium, but they lack explicit bounds.

\subsection{Results}
For any graded poset, let $\Omega_k$ be the elements of rank $k$ and let
$\Omega=\bigcup_{i=0}^n \Omega_i$ be the entire poset.  We show that {\it
provably efficient} {Boltzmann samplers} on $\Omega_k$ can be easily
constructed from certain {\it rapidly mixing Markov chains} on the Hasse
diagram of the entire poset $\Omega$, under very mild conditions.  We apply
this technique to design the first provably efficient Markov chain based
algorithms for sampling integer partitions of an integer $n$, permutations
with a fixed number of inversions, and lozenge tilings with fixed average
height.  Unlike all other methods for sampling that depend on
efficient counting techniques, our results extend to interesting subspaces of these
posets, such as partitions with at least $k$ pieces with size greater than~$\ell$,
or partitions into pieces with distinct sizes, or many other such
restricted classes.  For these subspaces, our results provide the first
sampling algorithms that do not require the space-expensive task of counting.
%
%

We focus on the example of integer partitions of $n$ and prove that there is
a Markov chain Monte Carlo algorithm for uniformly sampling partitions
of $n$ from a large family of \textit{region-restricted partitions},
i.e., Young diagrams restricted to any simply-connected bounding region.
The Markov chain on the Hasse diagram for partitions is the natural
``mountain-valley''  chain studied for staircase walks, tilings, and
permutations.  The transition probabilities are designed to generate a diagram
$\sigma$ with weight proportional to $\lambda^{\abs{\sigma}}$.  Previous work
on biased card shuffling \cite{bbhm} and growth
processes~\cite{bbhm,GPR09,LPW} shows that this chain is rapidly mixing for
any constant $\lambda$ on well-behaved regions.

In the general setting of sampling from a graded poset, our algorithm
is similar to current Boltzmann samplers that heuristically sample
elements of a given size, but often without rigorous analysis.  However, we
establish conditions under which these algorithms can be shown to be
efficient, including restricted settings for which no other methods provide
guarantees on both efficiency and accuracy.  
For example, we show that our method can produce random partitions of $n$ in
$O(n^{9/4})$ expected time with only $O(n^{1/2}\log(n))$ space.  Using {\it
coupling from the past}, we can in fact generate samples of the desired size
exactly uniformly, if this is desirable.

Although our algorithm is slower than recent results for sampling unrestricted
partitions using independent geometric random variables \cite{ad, fristedt}
(in the settings where those methods apply), our method is significantly more
versatile.  The Markov chain algorithm readily adapts to various restricted
state spaces, such as sampling partitions with bounded size and numbers of
parts, partitions with bounded Durfee square, and partitions with prescribed
gaps between successive pieces including partitions into pieces with distinct
sizes.
%
For general bounding regions, our algorithm still uses $O(n^{1/2}\log(n))$ space,
and hence is usually much more
suitable than other approaches with substantially larger space requirements.

Finally, we achieve similar results for sampling from fixed a rank in other
graded posets.  These include  permutations with a fixed number of inversions
and lozenge tilings with a given average height, referring to the height
function representation of the tilings (see, e.g., \cite{LRS01}).  Kenyon and
Okounkov \cite{ko} explored limit shapes of tilings with fixed volume, and
showed such constraints simplified some arguments, but there has not been work
addressing sampling.

\subsection{Techniques}
First, we present a new argument that shows how to build Boltzmann samplers with performance guarantees, 
even in cases where
the underlying distributions are not known (or necessarily even believed) to be unimodal, provided
the Markov chain is rapidly mixing on the whole Hasse diagram. 
We prove that there must be a 
{\it balanced bias} parameter $\lambda$ that we can find efficiently  allowing us to
generate configurations of the target size with probability at least
$1/{\rm{poly}}(n)$.   The desired set is no longer guaranteed to be the
mode of the distribution, as generally required, but we still show that rejection probabilities will
not be too high.  We carefully define a polynomial sized set 
from which the bias parameter $\lambda$ will be chosen.  Then we show that at
least one bias parameter in this set will define a distribution satisfying~$\sum_{i\leq k}
\Pr[\Omega_i] \geq 1/c$ and $\sum_{i>k} \Pr[\Omega_i] \geq 1/c,$ for some constant
$c$.  Because the Markov chain~$\m$ changes the rank by at most 1 in each step, we
must generate samples of size exactly~$k$ with  probability at least
$1/\tau(\m)$, where $\tau(\m)$ is the mixing time of~$\m$, which we prove
using conductance. Thus, when the
chain is rapidly mixing,  samples of size~$k$ must occur with
non-negligible probability.  
This new method based on {\it
balanced biases}  is quite general and  circumvents the need to
make any assumptions about the underlying distributions.

We use biased Markov chains and Boltzmann sampling to generate
samples of the desired size $k$.
We assign
probability $\lambda^{r(\sigma)}/Z$ to every element $\sigma \in \Omega$,
where $r(\sigma)$ is its rank and~$Z$ is
the normalizing constant.  When the underlying distributions on $f(i) =
|\Omega_i| \lambda^i$ are known to be log-concave in $i$, such as unrestricted
integer partitions or permutations with a fixed number of inversions,
we can provide better guarantees than the general balanced bias 
algorithm.

Several observations allow us to improve the running time of our algorithm, especially in the
case of unrestricted integer partitions.  First, instead of
sampling Young diagrams in an $n \times n$ lattice region, we restrict to
diagrams lying in the first quadrant of $\Z^2$ below the curve $y=2n/x$,
since this region contains all the Young diagrams of interest and has
area~$\Theta(n\log(n)),$ allowing the Markov chain to converge faster.
Next, we improve the bounds on the mixing time for our particular choice of $\lambda$
given in \cite{GPR09}
using a careful analysis of a recent result in \cite{LPW}.
Last, we show how to \emph{salvage} many of the samples
rejected by Boltzmann sampling to increase the success probability to at least $\Omega(1/n^{1/4})$.  
With all of these improvements
we conclude that the chain will converge in $O(n^2)$ time and $O(n^{1/4})$
trials are needed in expectation before generating a sample corresponding to a
partition of~$n$.
We also optimize the space required to implement the Markov chain. All Young diagrams in the region $R$
have at most $O(n^{1/2})$ corners, so each diagram in stored in $O({n}^{1/2} \log(n))$ space.  
\section{Bounding Rejection  with Balanced Bias}\label{sec:balanced-bias}

Let $\Omega$ be the elements of any graded
poset with rank function $r: {\Omega} \rightarrow \Z_{\geq 0}$.  The rank of
the poset $\Omega$ is $R = \max\parens*{\{r(\sigma) : \sigma \in \Omega\}}$ and the rank
generating function of $\Omega$ is
\[
  F_{\Omega}(x) = \sum_{\sigma \in \Omega} x^{r(\sigma)}.
\]
Let $\Omega_k$ be the set of elements of $\Omega$ with rank
$k$ and let $a_{\Omega,k}=\abs{\Omega_k}$.  For any $\lambda > 0$, the Gibbs
measure of each $\sigma \in \Omega$ is $\pi(\sigma) = \lambda^{r(\sigma)}/Z$,
where
\[
  Z = F_{\Omega}(\lambda) = \sum_{i=0}^{R} a_{\Omega,i} \lambda^i
\]
is the normalizing constant.

We define the natural Markov chain $\m$ that traverses the Hasse diagram of
$\Omega$ as follows. Let $\Delta$ be the maximum number of neighbors of any element
$\sigma \in \Omega$ in the Hasse diagram. For any pair of neighboring elements
$\sigma, \rho \in \Omega$, we define the transition probabilities
\[
  P(\sigma, \rho) = \frac{\min(1, \pi(\rho)/\pi(\sigma))}{2\Delta},
\]
and with all remaining probability we
stay at $\sigma$.  This Markov chain is known as the \emph{lazy,
Metropolis-Hastings algorithm}~\cite{metropolis} with respect to the Boltzmann distribution
$\pi(\sigma) = \lambda^{r(\sigma)}/Z$.
If $\m$ connects the state space of the
poset, the process $\sigma_t$ is guaranteed to converge to the stationary
distribution $\pi$ starting from any initial $\sigma_0$~\cite{LPW}.

The number of steps needed for the Markov chain $\m$ with state space $\Omega$
to get arbitrarily close to this stationary distribution is known as its
\emph{mixing time} $\tau(\varepsilon)$,
defined as $${\tau(\varepsilon) = \min (\{t : \|P^{t'},\pi\|_{tv} \leq
\varepsilon \text{ for all } t' \geq t\})},$$ for all $\varepsilon>0,$ where
$\|\cdot, \cdot\|_{tv}$ is the {\it total variation distance} (see, e.g.,
\cite{sinc}).   We say that a Markov chain is \emph{rapidly mixing} if the
mixing time is bounded above by a polynomial in $n$ and
$\log(\varepsilon^{-1})$.


We wish to uniformly sample a random element $\sigma \in \Omega_k$, for a fixed
$k \in [R]$. To achieve this, we repeatedly sample from a favorable Boltzmann
distribution over all of $\Omega$ until we have an element of rank $k$.
We show that under very mild conditions on the coefficients of
the rank generating function, it is sufficient that the Markov chain $\m$ over
$\Omega$ be rapidly mixing in order for the Boltzmann sampling procedure to be
efficient.  Specifically, we require only that $R = O(\text{poly}(n))$ and $1 \le
a_{\Omega,i} \le c(n)^i$ for some polynomial $c(n)$.

We formalize our claim by assuming the polynomial $c = c(n) \ge 2$.
For $t \ge 0$, 
let
\[
  \beta_t = \ln(1/c) + t \ln(c)/R
\]
and
\[
	\lambda_t = e^{\beta_t} = c^{t/R - 1}.
\]
Then let $\Pr_t[\sigma] = \lambda_{t}^{r(\sigma)} / Z_t$, 
  where $Z_t = F_{\Omega}(\lambda_t)$.
The sequence $\{\lambda_t\}_{t=0}^\infty$ is constructed in such a way
that at most $R^2$ values need to be considered.


\begin{lemma}\label{lemma:ratio}
For all $\sigma \in \Omega$ and $t \ge 0$, we have
\[
	\frac{\Pr_{t+1}[\sigma]}{\Pr_{t}[\sigma]} \ge \frac{1}{c}.
\]
\end{lemma}

\begin{proof}
By the definition of $\beta_{t+1}$, we have
\[
	1~\ge~\frac{e^{\beta_t r(\sigma)}}{e^{\beta_{t+1} r(\sigma)}}
  ~=~e^{-\ln(c)r(\sigma)/R}~\ge~\frac{1}{c}.
\]
It follows that
\begin{align*}
	\frac{\Pr_{t+1}[\sigma]}{\Pr_{t}[\sigma]}~&=~
	\frac{e^{\beta_{t+1}r(\sigma)}}{e^{\beta_{t}r(\sigma)}} \cdot \frac{Z_{t}}{Z_{t+1}}
	~\ge~\frac{Z_t}{Z_{t+1}}
	~=~ \frac{\sum_{\sigma \in \Omega} e^{\beta_{t}r(\sigma)} }{\sum_{\sigma \in \Omega} e^{\beta_{t+1}r(\sigma)}}
	~\ge~ \frac{\sum_{\sigma \in \Omega} e^{\beta_{t}(\sigma)}}{\sum_{\sigma \in \Omega} c e^{\beta_{t}r(\sigma)}}
  ~=~ \frac{1}{c}. \qedhere
\end{align*}
\end{proof}

The following lemma is critical to our argument and states that there exists
a {\it balanced bias} parameter $\lambda$ relative to our target set $\Omega_k$
that assigns nontrivial probability mass to elements with rank at most
$k$ and elements with rank greater than $k$.

\begin{lemma}\label{lemma:balanced-bias}
Let $\Omega$ be the elements of a graded poset with rank $R \ge 1$ such that
$1 \le a_{\Omega,i} \le c^i$ for all $i \in \{0,1,\dots,R\}$
and some $c \ge 2$.
If  $k \in [R-1]$, there exists a $t \in [R^2]$ for which
\[
	\Pr_{t}[r(\sigma) \le k] \ge \frac{1}{c+1}
\]
and
\[
 \Pr_{t}[r(\sigma) > k] \ge \frac{1}{c+1}.
\]
\end{lemma}

\begin{proof}
Suppose there exists a minimum $t^\ast \in \Z_{\ge 1}$ such that
\[
  \Pr_{t^\ast}[r(\sigma) > k] > \frac{1}{c+1}.
\]
Then
\[
  \Pr_{t^\ast-1}[r(\sigma)~\le~k] \ge \frac{c}{c+1},
\]
so by Lemma~\ref{lemma:ratio} we have
\[
  \Pr_{t^\ast}[r(\sigma) \le k] \ge \frac{1}{c+1}.
\]
To prove the existence of $t^\ast$, recall that
$
  \Pr_t[r(\sigma) > k] > 1/(c+1)
$
if and only if
\[
  (c+1)\sum_{i=k+1}^{R} a_{\Omega,i}\lambda_t^i~>~1.
\]
To prove the second inequality, it suffices to
show $(c+1)\lambda_t^{R} > 1$.  Letting $t=R^2$, we have $\lambda_t =
c^{R - 1} \ge 1$ because $R \ge 1$.  Therefore $(c+1)\lambda_t^{R} > 1$
as desired.  Finally, let $t^\ast$ be the minimum $t \in [R^2]$ satisfying
\[
  \Pr_{t}[r(\sigma) > k] > \frac{1}{c+1}. \qedhere
\]
\end{proof}

We now prove our main theorem, which depends on the mixing time $\tau(\varepsilon)$
of the Markov chain $\m$ for the balanced bias $\lambda_{t}$, given
by Lemma~\ref{lemma:balanced-bias}.  The proof uses a characterization of the
mixing time of a Markov chain in terms of its
\emph{conductance}~\cite{JS89,sinclair}.  For an ergodic Markov chain $\m$ with
stationary distribution~$\pi$, the {\it conductance} of a subset $S\subseteq
\Omega$ is defined~as
\[
  \Phi(S) = \sum_{\sigma \in S, \rho\in
  \overline{S}} \frac{\pi(\sigma)P(\sigma,\rho)}{\pi(S)}.
\]
The {\it conductance of the chain
$\m$} is the minimum conductance over all subsets
\[
  \Phi_\m = \min_{S\subseteq\Omega} \left(\{\Phi(S): { \pi(S) \leq 1/2}\}\right),
\]
and is related to the mixing time $\tau(\varepsilon)$ of $\m$ as
follows.

\begin{theorem}[\cite{JS89}]\label{conductance} 
The mixing time of a Markov chain $\m$ with conductance $\Phi$ satisfies
\[
  \tau(\varepsilon) \ \geq \ \left( \frac{1-2\Phi}{2\Phi} \right)
  \ln\left(\varepsilon^{-1}\right).
\]
\end{theorem}


\begin{theorem}[Balanced Bias]\label{crossover_thm}
Let $\m$ be a rapidly mixing Markov chain with state space~$\Omega$ and mixing
time $\tau = \tau(e^{-1})$ such that the transitions of $\m$ induce a graded
partial order on~$\Omega$ with rank function $r : \Omega \rightarrow \Z_{\ge
0}$ and rank $R$. If there exists a polynomial $c \ge 2$ such that $1 \le a_{\Omega,i}
\le c^i$ for all $i \in \{0,1,\dots,R\}$, then
\[
  \pi(\Omega_k)~\geq~\frac{1}{2(c+1)(\tau+1)}
\]
for any fixed $k \in [R]$ with the 
balanced bias.  If $\m$ can
be used to generate exact samples from~$\pi$ in expected $O(\tau)$ time, then
we can uniformly sample from $\Omega_k$ in expected $O(c\tau^2)$~time.
\end{theorem}

\begin{proof}
  Let $\m$ have conductance $\Phi$ and assume $k < R$.  
  Considering the cut $S = \Omega_{\le k}$
  and using Lemma~\ref{lemma:balanced-bias}, we have
  $\min(\pi(S),\pi(\overline{S})) \ge 1/(c+1)$
  for the balanced bias $\lambda_t$.
  It follows that
  \begin{align*}
    \Phi(S)
    ~\leq~ \frac{\sum_{\sigma\in S, \rho\in \overline{S}}
      \pi(\sigma)P(\sigma,\rho)}{\min(\pi(S), \pi(\overline{S}))}
    ~\leq~ (c+1){\sum_{\sigma\in \Omega_{k}, \rho\in \Omega_{k+1}}
      \pi(\sigma)P(\sigma,\rho)}
    ~\leq~ (c+1)~\pi(\Omega_k).
  \end{align*}
  By Theorem~\ref{conductance}, we have
  \[
    \Phi~\geq~\frac{1}{2(\tau + 1)},
  \]
  so
  \[
    \pi(\Omega_k)~\geq~\frac{1}{2(c+1)(\tau+1)}.
  \]
  It follows that $O(c\tau)$ samples from $\pi$ are needed in expectation to
  generate a uniform $\sigma \in \Omega_k$ for any fixed $k \in [R-1]$ with
  the given balanced bias.
  Moreover, if each sample is exactly generated in $O(\tau)$
  expected time, then the total running time of this sampling algorithm is
  $O(c\tau^2)$.
  The argument when $k=R-1$ extends to $k=R$ by the detailed balance equation.
\end{proof}

For simplicity, this theorem assumes we have a method for generating samples
{\it exactly} from $\pi$.  In many graded posets, including all considered
here, we can use the \emph{coupling from the past} algorithm to generate
perfect samples in expected $O(\tau)$ steps per sample~\cite{pw}.  In cases
when we cannot sample exactly, we have the following corollary of
Theorem~\ref{crossover_thm} that only requires samples be chosen close to $\pi$.

\begin{corollary}
  We can use $\m$ to
  approximately generate samples from $\Omega_k$ to within
  $\varepsilon$ of the total variation distance of $\pi$
  in expected 
  $O(c \tau^2 \max(\log(\varepsilon^{-1}), \log(c \tau)))$ time.
\end{corollary}
    
\begin{proof}
  Let
  \[
    \varepsilon^\ast = \min\left(\varepsilon,~\frac{1}{8(c+1)(\tau+1)}\right)
  \]
  be
  the desired bound on the total variation distance between the $t$-step
  distribution $P^t(\sigma,\cdot)$ starting from any initial $\sigma \in
  \Omega$ and the stationary distribution $\pi$.
  Then
  \begin{align*}
    \varepsilon^*~\ge~\frac{1}{2}\sum_{\rho \in \Omega}\abs{P^t(\sigma,\rho) - \pi(\rho)}
    ~\ge~\frac{1}{2} \abs{P^t(\sigma,\Omega_k) - \pi(\Omega_k)}.
  \end{align*}
  Theorem~\ref{crossover_thm} and our choice of $\varepsilon^*$ imply that
  \[
    P^t(\sigma,\Omega_k)~\ge~\pi(\Omega_k) - 2\varepsilon^*
    ~\ge~\frac{1}{4(c+1)(\tau+1)}.
  \]
  Each sample can be generated in $O(\tau(\varepsilon)) =
  O(\tau\log(\varepsilon^{-1}))$ steps, so the expected runtime is $O(c\tau^2
  \max(\log(\varepsilon^{-1}), \log(c\tau)))$.
\end{proof}
\section{Sampling Integer Partitions}\label{sec:background}

We demonstrate how to use the balanced bias technique to sample
from general classes of restricted integer partitions.
Integer partitions have a natural representation as \emph{Young diagrams},
which formally are finite subsets $\sigma \subseteq \Z_{\ge
0}^2$ with the property that if $(a,b) \in \sigma$,  then
\[
  \set{(x,y) \in \Z^2_{\ge 0} : 0 \leq x \leq a \text{ and } 0 \leq y \leq b} \subseteq \sigma.
\]
Young diagrams can be visualized as a connected set of unit squares on the
integer lattice with a corner at $(0,0)$ and a nonincreasing upper boundary
from left to right.  Each square in the Young diagram must be supported below
by the $x$-axis or another square and supported to the left by the $y$-axis or
another square.
We are interested in \emph{region-restricted Young diagrams}, a variant of
Young diagrams whose squares are restricted to lie in a connected region $R
\subseteq \Z_{\ge 0}^2$ such that each square is supported
below and to the left by the boundary of $R$ or another square.
Note that we use~$R$ in this section to denote a region instead of the rank of a
poset. We will see that the rank of the poset induced by the natural partial
order on $R$-restricted Young diagrams is $|R|$.

We call Young diagrams $\sigma \subseteq \Z_{\ge 0}^2$ such that
$\abs{\sigma}=n$ \emph{unrestricted integer partitions} of~$n$ and use this
term interchangeably with integer partitions.  Many well-studied classes of
restricted integer partitions have natural interpretations as region-restricted
Young diagrams. For example, the set of integer partitions of $n$ with at most
$k$ parts and with each part at most size $\ell$ give rise to the Gaussian
binomial coefficients and can be thought of as the set of Young diagrams of
size $n$ contained in a $k \times \ell$ box.

\begin{figure}[H]
\centering
\includegraphics[width=0.45\linewidth]{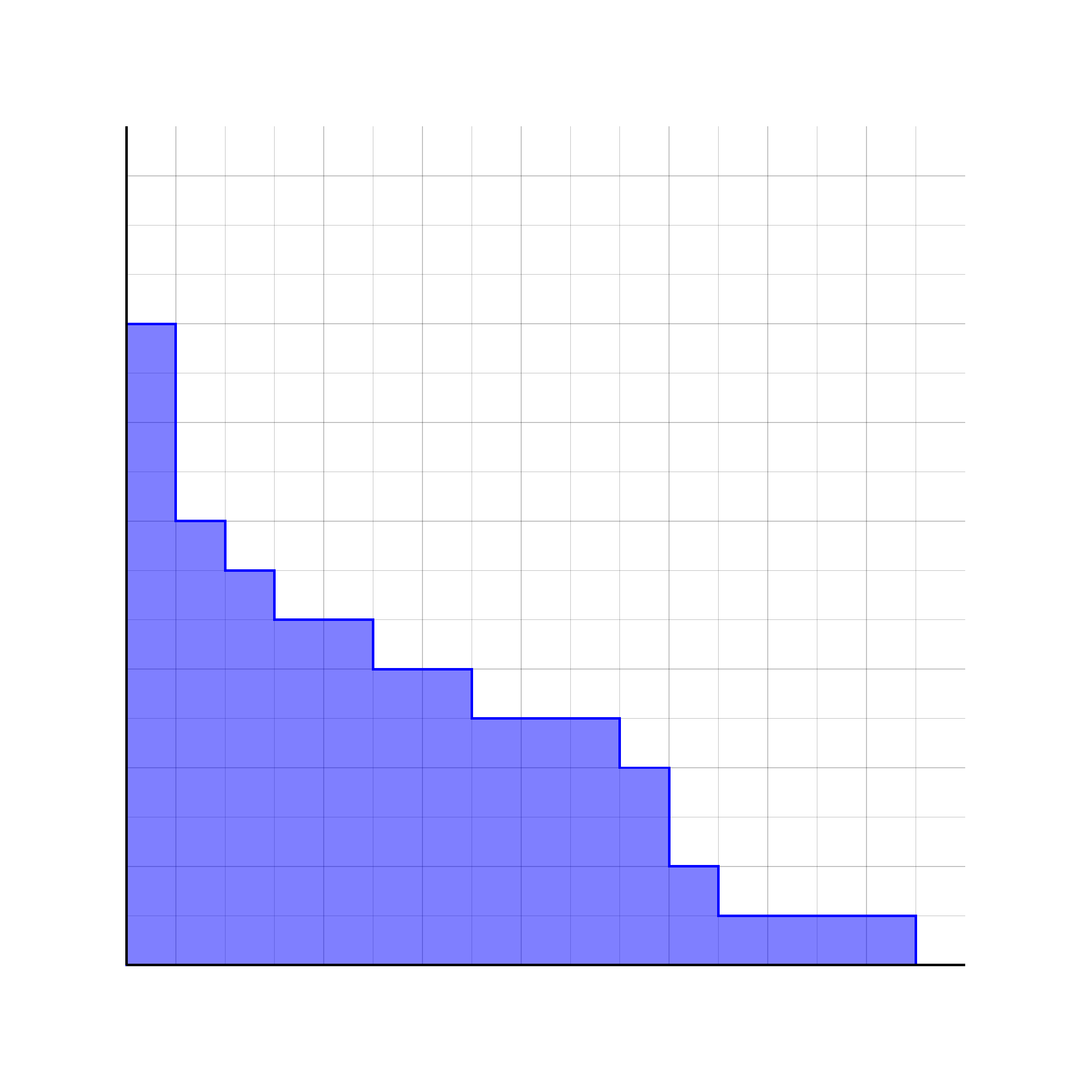}\hspace{0.5cm}
\includegraphics[width=0.45\linewidth]{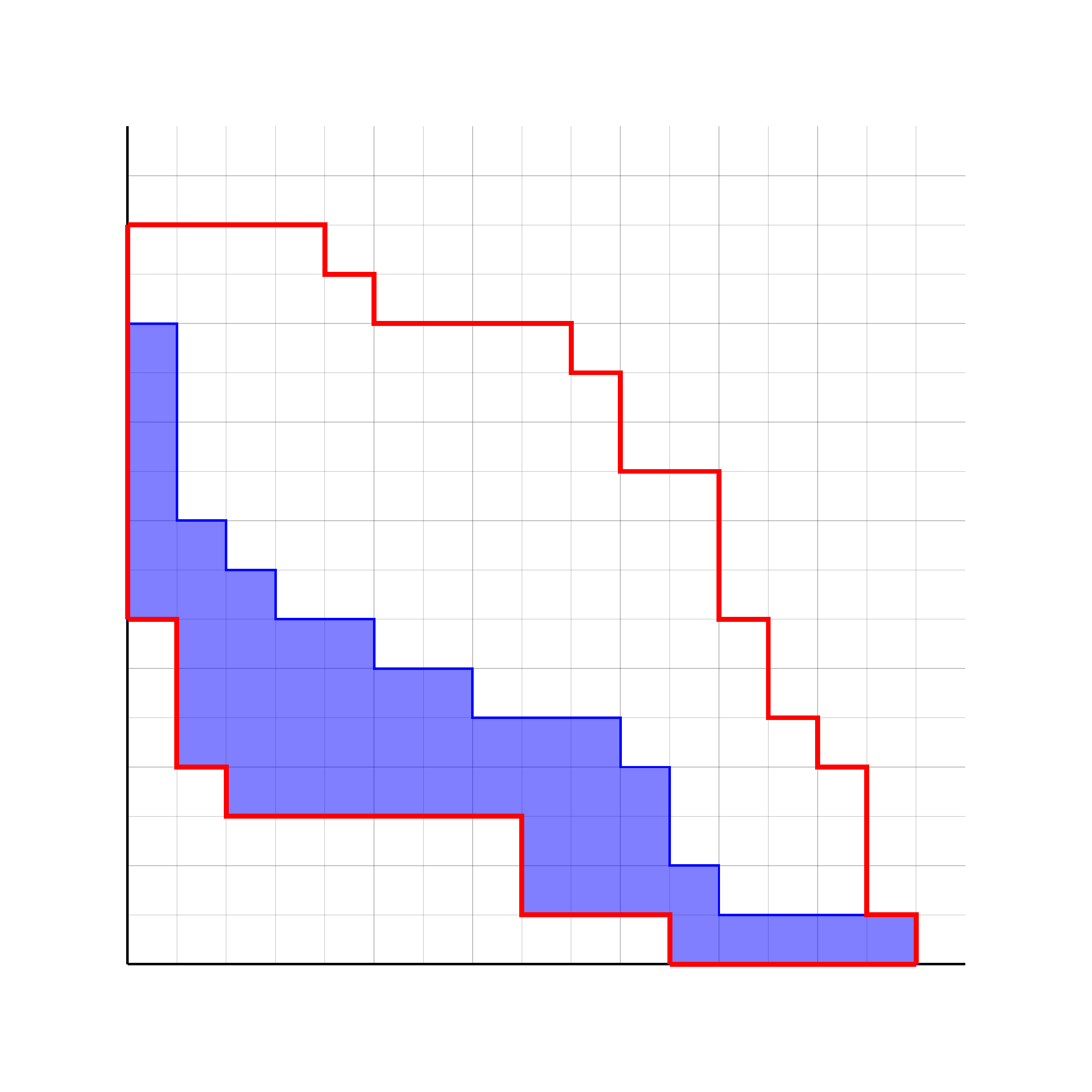}
\caption{Unrestricted and restricted integer partitions.}
\label{fig:partitions}
\end{figure}

\subsection{The Biased Markov Chain}
Let the state space $\Omega$ be the set of all Young diagrams
restricted to lie in a region $R$.
Young diagrams have a natural graded partial order via inclusion,
where $\sigma \leq \rho$ if and only if $\sigma
\subseteq \rho$, so the rank of a diagram $\sigma$ is
$r(\sigma) = \abs{\sigma}$.
The following Markov chain $\m$ on the Hasse diagram of this partial order
makes transitions that add or remove a square on the boundary of the diagram in each step
according to the Metropolis-Hastings algorithm. Therefore the stationary
distribution is a Boltzmann distribution parameterized by a bias value $\lambda$.
Let $R$ be a region such that every partition restricted to this region has at
most $\Delta$ neighboring configurations. 

\vspace{0.10in}
\noindent \ul{\textbf{Biased Markov Chain on Integer Partitions $\m$}}\label{mc}

\vspace{0.05in}
\noindent Starting at any Young diagram $\sigma_0 \subseteq R$, repeat:
\begin{itemize}
  \item Choose a neighbor $\rho$ of $\sigma_t$ uniformly at random with
  probability $1/2\Delta$.
  \item Set $\sigma_{t+1} = \rho$ with probability $\min(1, \lambda^{|\rho|-|\sigma_t|})$.
  \item With all remaining probability, set $\sigma_{t+1} = \sigma_t$.
\end{itemize}

The state space $\Omega$ is connected,
because any configuration can eventually reach the minimum configuration
$\sigma = \emptyset$ with positive probability.
By construction, $\m$ is lazy (i.e., it is always possible that $\sigma_t = \sigma_{t+1}$),
so it follows that $\m$ is an \textit{ergodic} Markov chain, and hence has
a unique stationary distribution $\pi$.
Using the detailed balance equation for Markov chains \cite{RM03}, we see that
$\pi(\sigma) = \lambda^{|\sigma|}$, for all $\sigma \in \Omega$.


This Markov chain can be used to efficiently approximate the number of
partitions of~$n$ restricted to $R$ within arbitrarily small
specified relative error, because this problem is
\emph{self-reducible} \cite{jerrum2003counting}.
Observe that we can run $\m$ restricted to $R$ polynomially many times
and compute the mean height $m$ in the first column of the sampled Young diagrams.
Then we use $\m$ to recursively approximate the number of partitions of $n-m$ restricted to
the region
\[
  R' = \{(x,y) \in R : 1 \le x \text{ and } y \le m \},
\]
and return the product of $m$ and this approximation.



\subsection{Sampling Using Balanced Bias}
In the following general sampling theorem for restricted integer partitions,
the mixing time of $\m$ must hold for all bias parameters $\lambda_t$.

\begin{theorem}\label{sampling_thm}
  Let $\tau = \tau(e^{-1})$ be the mixing time of $\m$ on the region $R$.
  We can uniformly sample partitions of $k$ restricted to a region $R$ in
  expected $O(\Delta \tau^2)$ time.
\end{theorem}
\begin{proof}
  There is only one such partition when $k=0$ or $k=|R|$, so assume $k \in [|R|-1]$.
  By construction $\vert\Omega_{k+1}\vert /
  \vert\Omega_{k}\vert \leq \Delta$ for all fixed $k$, so
  $1 \le \vert\Omega_i\vert \le \Delta^i$ for all $i \in \{0,1,\dots,|R|\}$.
  By Lemma~\ref{lemma:balanced-bias} there exists a balanced bias $\lambda$,
  which we can identify adaptively
  in $O(\log(|R|))$ time with a binary search as we are sampling,
  since Boltzmann distributions increase monotonically with increasing $\lambda_t$.
  Therefore, we can generate Young diagrams restricted to $R$ with any
  fixed rank in expected $O(\Delta\tau^2)$ steps of $\m$
  by Theorem~\ref{crossover_thm}.
\end{proof}

If more is known about the number of elements at each rank or the geometry of
$R$, then we can give better bounds on the runtime of this algorithm.  For
example, if $R$ is the region of a \emph{skew Young diagram} (see
Figure~\ref{fig:partitions}), a region contained between two Young diagrams,
then we can adapt Levin and Peres' mixing results about biased
exclusion processes to this setting.

\begin{theorem}[\cite{lp-be-16}]\label{thm:biased-exclusion}
Consider the biased exclusion process with bias $\beta = \beta_n = 2p_n - 1 > 0$
on the segment of length $2n$ and with $n$ particles.
Set $\alpha = \sqrt{p_n / (1 - p_n)}$.
For $\varepsilon > 0$, if $n$ is large enough, then
\[
  \tau(\varepsilon)
    \le \frac{4n}{\beta^2}\left[\log\left(\varepsilon^{-1}\right)
      + \log\left[\alpha\left(\frac{\alpha^n -1}{\alpha-1}\right)^2 \right] \right].
\]
\end{theorem}

\begin{corollary}
  If the region $R$ is a skew Young diagram contained in an $n \times n$ box,
  we can uniformly sample partitions of $k$ restricted to $R$ in expected
  $O(n^{16})$ time.
\end{corollary}

\begin{proof}
  The biased exclusion process on a segment of length $2n$ with $n$ particles
  is in bijection with $\m$ when the restricting region is an $n \times n$ box.
  The proof of Theorem~\ref{thm:biased-exclusion} in \cite{lp-be-16} uses a path
  coupling argument that directly extends to and gives an upper bound for the
  mixing time of $\m$ when the region $R$ is a skew Young diagram, since the
  expected change in distance of two adjacent states in the more restricted
  setting can only decrease.
  Let $\lambda_{n,t}$ denote $\lambda_t$ in an instance of size $n$.
  We analyze the three cases $\lambda_{n,t} < 1$, $\lambda_{n,t} = 1$, and
  $\lambda_{n,t} > 1$, and then bound the mixing time of $\m$ for all $\lambda_{n,t}$.

  To prove the existence of a balanced bias using
  Lemma~\ref{lemma:balanced-bias}, observe that ${a_{\Omega, i} \le p(i)
  \le 2^i}$ for all $i \in \{0, 1, \dots, n^2\}$.
  In the first case, assume $\lambda_{t,n} < 1$. Then we have $t \in [n^2 - 1]$
  since $|R| = n^2$.
  Translating $\m$ to the biased exclusion process,
  \[
    p_{n,t} = \frac{1}{1 + \lambda_{n,t}},
  \]
  \[
    \beta_{n,t} = \frac{1 - \lambda_{n,t}}{1+\lambda_{n,t}},
  \]
  and
  \[
    \alpha_n = \sqrt{\frac{1}{\lambda_{n,t}}}.
  \]
  To use Theorem~\ref{thm:biased-exclusion},
  we first prove $1/\beta_{n,t}^2 \le 10n^4$.
  To see this, observe
  that $t=n^2-1$ minimizes $\beta_{n,t}$, hence maximizes the desired quantity.
  Then
  \[
    \lim_{n\rightarrow\infty} \left(\frac{1}{\beta_{n,n^2-1}}\right)^2 = \frac{4n^2}{\log(2)^2} \le 10n^4.
  \]
  Next, since $a_{n,t} > 1$, we have
  \begin{align*}
    \alpha_{n,t}\left(\frac{\alpha_{n,t}^n - 1}{\alpha_{n,t} - 1}\right)^2
    \le~\alpha_{n,t} \left(n\alpha_{n,t}^n \right)^2
    ~\le~\frac{n^2}{\lambda_{n,t}^{n+1/2}}~\le~n^2 2^{n + 1/2},
  \end{align*}
  because $\lambda_{n,t} \ge 1/2$.
  Thus, 
  $\tau(\varepsilon) = O(n^5(\log(\varepsilon^{-1}) + n))$
  for all $\lambda_{n,t}$ by Theorem~\ref{thm:biased-exclusion}.

  In the unbiased case when $\lambda_{n,t} = 1$, Wilson~\cite{wilson} proved that
  the mixing time of
  $\m$ is~$\Theta(n^3 \log(n/\varepsilon))$.
  In the third case, $\lambda_{t,n} > 1$ so
  $t \in \{n^2+1, n^2+2,\dots, n^4\}$,
  \[
    p_{n,t} = \frac{\lambda_{n,t}}{1+\lambda_{n,t}},
  \]
  \[
    \beta_{n,t} = \frac{\lambda_{n,t}-1}{1+\lambda_{n,t}},
  \]
  and
  \[
    \alpha_{n,t} = \sqrt{\lambda_{n,t}}.
  \]
  By similar analysis, $1/\beta_{n,t}^2 \le 10n^4$ and
  \begin{align*}
    \alpha_{n,t}\left(\frac{\alpha_{n,t}^n - 1}{\alpha_{n,t} - 1}\right)^2
    \le~n^2 \lambda_{n,t}^{n+1/2}~\le~n^2 \left(2^{n^2-1}\right)^{n+1/2},
  \end{align*}
  since $\lambda_{n,t} \le 2^{n^2 - 1}$.
  Thus $\tau(\varepsilon) = O(n^5(\log(\varepsilon^{-1}) + n^3))$,
  so by Theorem~\ref{crossover_thm} we can uniformly sample
  partitions of $k$ restricted to $R$ in expected $O(n^{16})$ time.
\end{proof}
\subsection{Sampling Using Log-concavity}\label{sec:truncation}

When more is known about the stationary distribution $\pi$, specifically the
sequence $\{|\Omega_i|\}_{i=0}^\infty$, we can typically improve the bounds on the running
time of our algorithm. In particular, we show that we can sample unrestricted
integer partitions in expected $O(n^{9/4})$ time.  Our primary techniques
involve using a compressed representation of partitions and using log-concavity
to show strong probability concentration around partitions of the desired
size.  These techniques extend to a variety of settings where log-concavity or
probability concentration can be shown.

To sample integer partitions of $n$, we set the bias parameter
$\lambda_n = p(n-1)/p(n)$ to force the
stationary distribution to concentrate at~$n$.
The sequence $\{p(k)\}_{k=26}^\infty$ is log-concave~\cite{dp, jln},
so it follows that the sequence
$\{p(k) \lambda_n^k \}_{k=26}^\infty$ is, too.
Log-concave sequences of positive terms are unimodal, which implies
that the mode of our stationary distribution is at $k=n$.
Moreover, we show how log-concavity gives exponential decay on both
sides of the mode, and hence strong concentration.


We now argue that we need only consider Young diagrams that lie under the curve
$y=2n/x$ to sample partitions of $n$, as all Young diagrams with squares
above that curve must have more than $2n$ squares total.

\begin{proposition}\label{claim:space}
A Young diagram that lies under the curve $y=2n/x$ 
can be stored in $O(n^{1/2} \log(n))$ space.
\end{proposition}

\begin{proof}
For any square in the Young diagram, both of its coordinates are not greater than~$\sqrt{2n}$, for then it would lie above $y=2n/x$.  We may record the height of
each column and the width of each row in the range $\setto[0,1]{{\floor{\sqrt{2n}}
- 1}}$ to capture the position of every square in the diagram.
Therefore, we can represent the diagram using exactly these $2\floor{\sqrt{2n}}$
heights and widths.
\end{proof}

Using the compressed representation in the previous proposition, we see that
there will not be more than $O(n^{1/2})$ possible transitions at any possible
state, since our algorithms adds or removes at most one square on the upper
boundary in each step.  Note that we can adapt this technique in the
general case for any region $R$ that lies under the curve $y = 2n/x$.

\begin{proposition}\label{claim:moves}
There are at most $4\sqrt{2n}$ potential transitions for any Young diagram
that lies under the curve $y=2n/x$.
\end{proposition}

\begin{proof}
Observe that since the squares in any row or column must be connected, there
are at most two valid moves in any particular row or column.
Therefore, by Proposition~\ref{claim:space}, there are at most $4\floor{\sqrt{2n}}$ possible
transitions from any such Young diagram.
\end{proof}

We now shift our attention to bounding $\lambda_n$ and the consequences it has on
both the mixing time of $\m$ and the concentration of $\pi$.
Hardy and Ramanujan~\cite{hr} gave the classical asymptotic formula for the partition numbers
\[
  p(n) \sim \frac{1}{4\sqrt{3}n} e^{\pi \sqrt{2n/3}},
\]
and we use related bounds given in $\cite{dp}$ for the following lemma.
The proof is deferred to the next subsection.


\begin{lemma}\label{lem:lambda-bound}
For all $n \ge 30$, we have
\[
  1 - \frac{2}{\sqrt{n}}~<~\lambda_n~<~1- \frac{1}{\sqrt{n}}.
\]
\end{lemma}


\begin{theorem}\label{thm:mainmixing}
The Markov chain $\m$ with bias $\lambda_n$ restricted to
the region $R$ bounded by the curve $y = 2n/x$
mixes in $O(n^{3/2}(\log(\varepsilon^{-1}) + n^{1/2}))$.
\end{theorem}

\begin{proof}
We modify Theorem~\ref{thm:biased-exclusion} and its proof in \cite{lp-be-16}.
In this biased exclusion process, $\lambda = \lambda_n$,
$\beta = (1-\lambda)/(1+\lambda)$, and $\alpha = \sqrt{1/\lambda}$.
By Proposition~\ref{claim:moves}, there are at most $4\sqrt{2n}$ transitions
from any state, so for $n$ large enough
\[
  \tau(\varepsilon)
    ~\le~\frac{8\sqrt{2n}}{\beta^2}\left[\log\left(\varepsilon^{-1}\right)
      + \log\left(\text{diam}\left(\Omega\right)\right) \right],
\]
where $\text{diam}(\Omega)$ is the maximum length path between any two states,
as defined in \cite{LPW}.
Therefore, we have $\text{diam}(\Omega) \le |R|\alpha^{2n}$ and
\[
  |R|~\le~2n H_{2n}~\le~2n(\log(2n)+1),
\]
so
\[
  \text{diam}(\Omega)~\le~2n(\log(2n)+1)\alpha^{2n}~=~2n(\log(2n)+1)\lambda^{-n}.
\]
By Lemma~\ref{lem:lambda-bound} and the bound $1+x \le e^x$, for all $x \in \R$,
\begin{align*}
  \log\left(\lambda^{-n}\right)
  &\le \log\left(\left(1+\frac{2}{\sqrt{n}-2}\right)^n\right)
  \le 3\sqrt{n},
\end{align*}
for $n$ sufficiently large.
We have
\[
  \frac{1}{\beta}~\le~2\sqrt{n} - 1
\]
by Lemma~\ref{lem:lambda-bound}.
Therefore, $\tau(\varepsilon) = O(n^{3/2}(\log(\varepsilon^{-1}) + n^{1/2}))$.
\end{proof}

Another key observation we make to generate partitions of $n$ more efficiently
is to salvage samples larger than $n$ instead of rejecting them,
while preserving uniformity on the distribution $\Omega_n$.
For any $k \ge 0$, consider the function $f_k : \Omega_n \rightarrow
\Omega_{n+k}$ that maps a partition $\sigma=(\sigma_1,\sigma_2,\dots,\sigma_m)$
to $f_k(\sigma) = (\sigma_1+k,\sigma_2,\dots,\sigma_m)$.
Note that $\sigma_1 \ge \sigma_2 \ge \dots \ge \sigma_m$ since $\sigma$ is a Young diagram.
Clearly $f_k$ is injective, so we
can consider the inverse map $f^{-1}_k((\rho_1,\rho_2,\dots,\rho_\ell))$ that subtracts
$k$ from $\rho_1$ if $\rho_1 - k \ge \rho_2$, and is invalid otherwise.  Then, define
$g : \Omega_{\ge n} \rightarrow \Omega_n\cup\set{0}$ as 
\begin{align*}
  g((\rho_1,\rho_2,\dots,\rho_\ell)) = \begin{cases}
  (\rho_1-k,\rho_2,\dots,\rho_\ell) &
    \text{if $\rho_1+\rho_2+\dots+\rho_\ell = n+k$ and $\rho_1 - k
\ge \rho_2$} \\
  0 & \text{otherwise}.
\end{cases}
\end{align*}
The following lemma, whose proof is deferred to the next
subsection, uses the log-concavity of the partition numbers to
give a strong lower bound on the success of the map $g$.

\begin{lemma}\label{lem:rejection}
Let $\sigma$ be a random Young diagram from the stationary distribution of $\m$, and let
$g$ be the function defined above.  Then for all $n$ sufficiently large,
\[
  \Pr[g(\sigma) \text{~generates a partition of } n] \ge \frac{1}{160n^{1/4}}.
\]
\end{lemma}

Assembling the ideas in this section, we now
formally present our Markov chain Monte Carlo algorithm for
generating partitions of $n$ uniformly at random.

\vspace{0.10in}
\noindent \underline{\textbf{Algorithm for Sampling Integer Partitions}}

\vspace{0.05in}
\noindent
Repeat until success:
\begin{itemize}
  \item Sample $\sigma \in \Omega$ using $\m$.
  \item If $n \le \abs{\sigma} \le 2n$ and $g(\sigma) \ne 0$, return $g(\sigma)$.
\end{itemize}

\noindent
Note that we restrict $\abs{\sigma} \le 2n$ instead of $\abs{\sigma} \le
2n\log(n)$ so that $g$ maps to $\Omega_n$ uniformly.  All partitions of $2n$
are elements of $\Omega_{2n}$, but the same is not true for larger partitions
since the bounding region $R$ is the curve $y = 2n/x$.
Lastly, recall that
coupling from the past can be used efficiently in this setting to generate
perfectly uniform samples, because the natural coupling is monotone and there
is a single minimum and maximum configuration \cite{GPR09}.

\begin{theorem}
Our Markov chain Monte Carlo algorithm for generating a uniformly
random partition of $n$ runs in expected $O(n^{9/4})$ time and $O(n^{1/2}\log(n))$
space.
\end{theorem}

\begin{proof}
  The proof directly follows from
  Proposition~\ref{claim:space},
  Lemma~\ref{lem:rejection},
  and Theorem~\ref{thm:mainmixing}.
\end{proof}

\subsection{Proofs of Lemma~\ref{lem:lambda-bound} and Lemma~\ref{lem:rejection}}\label{sec:appendix-proofs}
We prove Lemma~\ref{lem:lambda-bound} using bounds for $p(n)$ given in \cite{dp}.
Let
\[
  \mu(n) = \mu_n = \frac{\pi\sqrt{24n - 1}}{6},
\]
\[
  \nu(n) = \nu_n = \frac{\sqrt{12}}{24n-1},
\]
and
\[
		T(n) = \nu_n \left[\left(1 - \frac{1}{\mu_n}\right) 
			e^{\mu_n} + \frac{(-1)^n}{\sqrt{2}}e^{\mu_n/2}\right].
\]
The function $T(n)$ is the sum of the three largest terms in the Hardy-Ramanujan
formula, and the explicit error bounds in \cite{dp} that we use were first proved by Lehmer \cite{lehmer}.
We only prove upper bounds in the following two proofs, as the lower bounds
are proved similarly.

\begin{lemma}\label{lem:p_bound}
For all $n \ge 2$, we have
\begin{align*}
	\left| p(n) - \nu_n \left(1 - \frac{1}{\mu_n}\right)e^{\mu_n}\right|
	~<~1 + e^{\mu_n/2}.
\end{align*}
\end{lemma}

\begin{proof}
By Lemma 2.3 and Proposition 2.4 in \cite{dp},
\[
  p(n)~<~T(n) + 1 + \frac{16}{\mu_n^3}e^{\mu_n/2}
  ~<~ \nu_n \left(1 - \frac{1}{\mu_n}\right) e^\mu_n  + 1 + e^{\mu_n/2}. \qedhere
\]
\end{proof}

{
\renewcommand{\proofname}{Proof of Lemma~\ref{lem:lambda-bound}.}
\begin{proof}
By Lemma~\ref{lem:p_bound},
\begin{align*}
  \lambda_n
		< \frac{1 + e^{\mu_{n-1}/2} + \nu_{n-1} \left(1 - \frac{1}{\mu_{n-1}}\right) 
			e^{\mu_{n-1}}}
		{-\left(1 + e^{\mu_n/2}\right) + \nu_n \left(1 - \frac{1}{\mu_n}\right) 
			e^{\mu_n}}
	= 
    \frac{e^{\mu_{n-1}}}{e^{\mu_n}}
		 \left(\frac{ e^{-\mu_{n-1}} + e^{-\mu_{n-1}/2} + \nu_{n-1} \left(1 - \frac{1}{\mu_{n-1}}\right)}{
		 -\left(e^{-\mu_{n}} + e^{-\mu_{n}/2}\right) + \nu_n \left(1 - \frac{1}{\mu_n}\right)}\right),
\end{align*}
for all $n \ge 14$, because the lower bound for $p(n)$ is initially negative.
We have
\begin{align*}
	e^{-\mu_n} + e^{-\mu_n/2} < \frac{\nu_n}{n} \left(1 - \frac{1}{\mu_n}\right),
\end{align*}
for all $n \ge 65$, so it follows that
\begin{align*}
	\lambda_n ~<~ 
	\frac{e^{\mu_{n-1}}}{e^{\mu_n}}\left(\frac{\left(1 + \frac{1}{n-1} \right)\nu_{n-1} \left(1-\frac{1}{\mu_{n-1}}\right) }
	{\left(1-\frac{1}{n}\right) \nu_{n} \left(1-\frac{1}{\mu_n}\right)  } \right)
	= \frac{e^{\mu_{n-1}}}{e^{\mu_n}} \left(\frac{n}{n-1}\right)^2 \left(\frac{\mu_n^3\left(\mu_{n-1}-1\right)}{\mu_{n-1}^3\left(\mu_n-1\right)}\right),
\end{align*}
for all $n \ge 2$.
Observe that $\mu_{n-1} - \mu_n < -\pi/\sqrt{6n}$ and
\[
	\frac{\mu_n^3\left(\mu_{n-1}-1\right)}{\mu_{n-1}^3\left(\mu_n-1\right)}
	~<~ \frac{n}{n-1},
\]
for all $n \ge 2$.
Using $e^x \le 1+x+x^2/2$, for all $x \le 0$,
\begin{align*}
	\lambda_n ~<~ \frac{e^{\mu_{n-1}}}{e^{\mu_n}} \left(\frac{n}{n-1}\right)^3
						< e^{-\frac{\pi}{\sqrt{6n}}}\left(\frac{n}{n-1}\right)^3
						\le \left(1-\frac{\pi}{\sqrt{6n}}+\frac{\pi^2}{12n}\right)\left(\frac{n}{n-1}\right)^3 <~ 1 - \frac{1}{\sqrt{n}},
\end{align*}
where the final inequality is true for all $n \ge 160$.
When $30 \le n < 160$, we verify the claim numerically. \qedhere
\end{proof}
}


Now we prove Lemma~\ref{lem:rejection} by
showing that the truncation scheme $g(\sigma)$ succeeds with sufficient
probability.
By Hardy-Ramanujan formula,
we have that for any constant $c>0$ and $n$ sufficiently large,
\begin{align*}
  \frac{1-c}{4\sqrt{3}n} e^{\pi\sqrt{2n/3}} ~\le~ p(n) ~\le~
  \frac{1+c}{4\sqrt{3}n} e^{\pi\sqrt{2n/3}}.
\end{align*}
Letting $\lambda = \lambda_n$, 
the Hardy-Ramanujan formula implies
that for all $n \ge 20$,
\[
  e^{-\pi k / \sqrt{6n}} ~\le~ \lambda^k \le e^{k/n-\pi k
  / \sqrt{6n}}.
\]

\begin{lemma}\label{lemma:Z_n}
Let $Z_n$ be the normalizing constant of the desired distribution.
Then we have
\[
  Z_n < 40 n^{3/4} \lambda^n p(n),
\]
for all $n$ sufficiently large.
\end{lemma}

\begin{proof}
Clearly
\[
  Z_n \le \sum_{k=0}^\infty p(k)\lambda^k.
\]
We further know that $f(k) = p(k)\lambda^k$ is unimodal with a maximum at $k=n$.
By the log-concavity of $\{f(k)\}_{k =26}^\infty$, we have
\[
  \frac{f(n+k)}{f(n)} \ge \frac{f(n+2k)}{f(n+k)} 
\]
and
\[
 \frac{f(n-k)}{f(n)} \ge \frac{f(n-2k)}{f(n-k)},
\]
for all $k \ge 1$.
Therefore, we can bound $Z_n$ as 
\begin{align*}
  Z_n \le k f(n) \left(\frac{1}{1-\frac{f(n+k)}{f(n)}} 
          + \frac{1}{1-\frac{f(n-k)}{f(n)}}\right), 
\end{align*}
for any $k \ge 1$.
Specifically, if both $f(n+k)/f(n)$ and $f(n-k)/f(n)$ are at most some fixed constant less than~$1$, then $Z_n=O(k f(n))$.
Using the bounds above,
\begin{align*}
  f(k) ~=~ p(k)\lambda^{k}
  ~\le~ \frac{1+c}{4\sqrt{3}k} e^{k/n-\pi\left(k - 2\sqrt{kn}\right)/\sqrt{6n}}
  ~=~ \frac{1+c}{4\sqrt{3}k} e^{k/n - \pi\left(\sqrt{k}-\sqrt{n}\right)^2/\sqrt{6n} + \pi\sqrt{n/6}}.
\end{align*}
Letting $n+k = (\sqrt{n}+n^{1/4})^2$, for $n$ large enough,
\[
  f\left(\left(\sqrt{n}+n^{1/4}\right)^2\right)
    ~\le~ \frac{1+c}{4\sqrt{3}n} e^{1.1 - \pi/\sqrt{6} + \pi\sqrt{n/6}}.
\]
We can then bound the density value at $(\sqrt{n}+n^{1/4})^2$ relative to the
maximum by
\begin{align*}
  \frac{f\left(\left(\sqrt{n}+n^{1/4}\right)^2\right)}{f(n)}
  ~\le~ \frac{ \frac{1+c}{4\sqrt{3}n} e^{1.1 - \pi/\sqrt{6} + \pi\sqrt{n/6}} }{ \frac{1-c}{4\sqrt{3}n} e^{\pi\sqrt{2n/3} - \pi\sqrt{n/6}} }
~=~ \frac{1+c}{1-c} e^{1.1 - \pi/\sqrt{6}}.
\end{align*}
Taking $c \le 0.01$, we have
\[
  \frac{f\left(\left(\sqrt{n}+n^{1/4}\right)^2\right)}{f(n)} ~\le~
      \frac{1.01}{0.99}  e^{1.1-\frac{\pi}{\sqrt{6}}} ~<~ 0.85.
\]
Similarly, for $n$ sufficiently large,
\[
  \frac{f\left(\left(\sqrt{n}-n^{1/4}\right)^2\right)}{f(n)} ~<~ 0.85.
\]
Therefore, we have $Z_n < 40 n^{3/4} \lambda^n p(n)$ using the fact
that $k \le 3n^{3/4}$.
\end{proof}

{
\renewcommand{\proofname}{Proof of Lemma~\ref{lem:rejection}.}
\begin{proof}
We use Lemma~\ref{lemma:Z_n} and Lemma~\ref{lem:lambda-bound} to bound the
probability that~$g(\sigma)$ generates a partition of $n$ successfully.
Therefore, we have
\begin{align*}
  \Pr[g(\sigma) \text{ generates a partition of } n]
  &= \sum_{k=0}^{n} \frac{\lambda^{n+k}p(n)}{Z_n}\\
  &\ge \frac{1}{40n^{3/4}} \sum_{k=0}^{n} \lambda^k\\
  &\ge \frac{1}{40n^{3/4}} \sum_{k=0}^{n} \left(1-\frac{2}{\sqrt{n}}\right)^k\\
  &= \frac{1}{40n^{3/4}} \cdot \frac{\sqrt{n}}{2} \left(1 - \left(1 - \frac{2}{\sqrt{n}}\right)^{n+1}\right)\\
  &\ge \frac{1}{160n^{1/4}}. \qedhere
\end{align*}
\end{proof}
}

\section{Sampling in Other Graded Posets}
We demonstrate the versatility of using a Markov chain on the Hasse diagram of
a graded poset to sample elements of fixed rank.  When this chain is
rapidly mixing for all $\lambda_t$ with~${t \in [R^2]}$, we can apply Boltzmann
sampling and Theorem~\ref{crossover_thm} to generate approximately uniform samples in polynomial
time.
Similar to region-restricted integer partitions, analogous notions of
self-reducibility apply to restricted families of permutations and lozenge
tilings, so there exist fully polynomial-time approximation schemes for
these enumerations problems because we can efficiently sample elements of a given
rank from their respective posets~\cite{jerrum2003counting}.


\subsection{Permutations with Fixed Rank}
In the first case, we consider permutations of $n$ elements with a fixed number
of inversions.  The Hasse diagram in this setting connects permutations that
differ by one adjacent transposition.  This partial order is in bijection with
the \emph{weak Bruhat order} on the symmetric group.
In the unbiased case ($\lambda = 1$), the nearest neighbor Markov chain mixes
in time $\Theta(n^3 \log(n))$~\cite{wilson}.  With constant bias the chain is known
to converge in time $\Theta(n^2)$ \cite{bbhm, GPR09}.  The number of
permutations of $n$ with $k$ inversions is known to be log-concave in $k$, so
standard Boltzmann sampling techniques can be used.  However, using our balanced bias method,
we avoid the need for bounds on the growth of inversion numbers in
restricted settings, such as permutations where at least $i$ of the first
$j$ elements are in the first half of the permutation.
Figure~\ref{fig:inversions} illustrates the distribution of inversions
of random permutations in $S_{100}$ sampled from various ranks of the inversion
poset as Rothe diagrams~\cite{kerber2013applied}.

\begin{figure}[H]
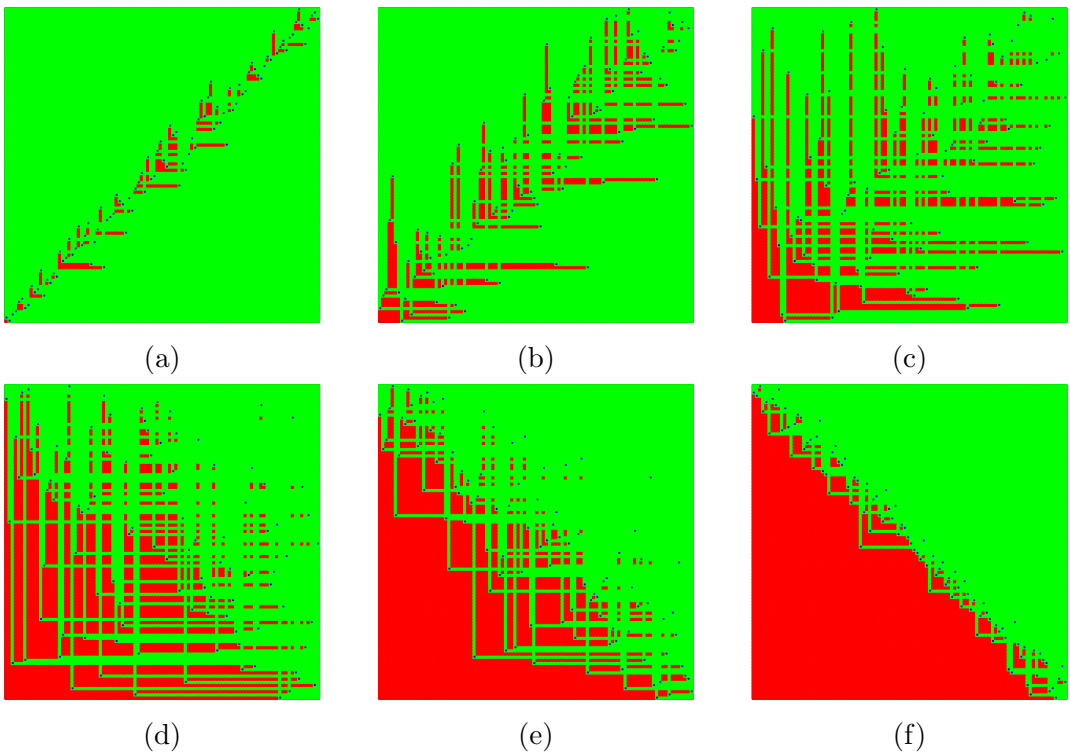

\centering
\begin{subfigure}[t]{0.26\textwidth}
  \includegraphics[width=1.0\textwidth]{images/rothe-005-percent}
  \caption{}
\end{subfigure}
\hspace{0.4cm}
\begin{subfigure}[t]{0.26\textwidth}
  \includegraphics[width=1.0\textwidth]{images/rothe-020-percent}
  \caption{}
\end{subfigure}
\hspace{0.4cm}
\begin{subfigure}[t]{0.26\textwidth}
  \includegraphics[width=1.0\textwidth]{images/rothe-040-percent}
  \caption{}
\end{subfigure}

\begin{subfigure}[t]{0.26\textwidth}
  \includegraphics[width=1.0\textwidth]{images/rothe-060-percent}
  \caption{}
\end{subfigure}
\hspace{0.4cm}
\begin{subfigure}[t]{0.26\textwidth}
  \includegraphics[width=1.0\textwidth]{images/rothe-080-percent}
  \caption{}
\end{subfigure}
\hspace{0.4cm}
\begin{subfigure}[t]{0.26\textwidth}
  \includegraphics[width=1.0\textwidth]{images/rothe-095-percent}
  \caption{}
\end{subfigure}
\caption{Random permutations with
(a)~5, (b)~20, (c)~40, (d)~60, (e)~80, (f)~95 percent of $\binom{100}{2}$ inversions.}
\label{fig:inversions}
\end{figure}

\subsection{Lozenge Tilings with Fixed Average Height}
Lozenge tilings are tilings of a triangular lattice region with pairs of
equilateral triangles that share an edge.  There is a well-studied height
function that maps hexagonal lozenge tilings bijectively to plane partitions
lying in an $n\times n\times n$ box (see, e.g., \cite{LRS01}), and it follows that lozenge
tilings with a fixed average height of $k$ are precisely the plane partitions
with volume~$k$.  The Markov chain that adds or removes single cubes on the
surface of the plane partition (corresponding to rotating three nested lozenges
180 degrees) is known to mix rapidly in the unbiased case.  Caputo et
al.~\cite{caputo} studied the biased version of this chain with a preference
toward removing cubes, and showed that this chain converges in $O(n^3)$ time.
Applying the balanced bias method, we can use Boltzmann sampling to generate random
lozenge tilings with any target average height in polynomial time, as shown in
Figure~\ref{fig:lozenge}.

\begin{figure}[H]
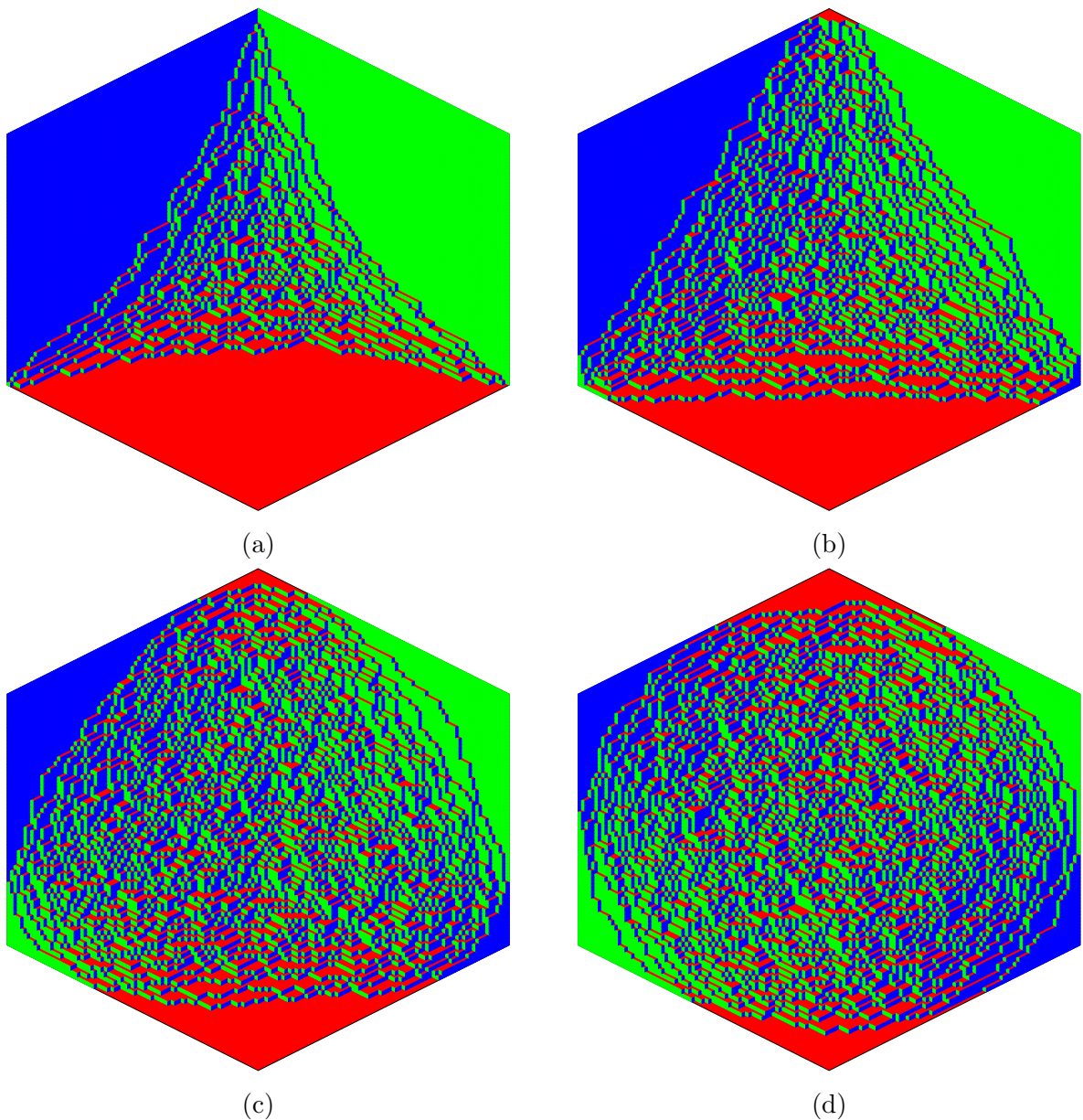

\centering
\begin{subfigure}[t]{0.45\textwidth}
  \includegraphics[width=1.0\textwidth]{images/lozenge-75-049}
  \caption{}
\end{subfigure}
\hspace{0.6cm}
\begin{subfigure}[t]{0.45\textwidth}
  \includegraphics[width=1.0\textwidth]{images/lozenge-75-144}
  \caption{}
\end{subfigure}

\begin{subfigure}[t]{0.45\textwidth}
  \includegraphics[width=1.0\textwidth]{images/lozenge-75-336}
  \caption{}
\end{subfigure}
\hspace{0.6cm}
\begin{subfigure}[t]{0.45\textwidth}
  \includegraphics[width=1.0\textwidth]{images/lozenge-75-490}
  \caption{}
\end{subfigure}
\caption{Random lozenge tilings with average height (a)~5, (b)~15, (c)~35, (d)~50 percent of~$75^3$.}
\label{fig:lozenge}
\end{figure}

\bibliographystyle{plain}
\bibliography{bib}


\end{document}